\documentclass[11pt]{article}
 \usepackage{helvet}
\usepackage{amsmath,amsthm,amssymb,setspace,enumitem,epsfig,rotating,titlesec,verbatim,color,array
}
\usepackage[round, colon]{natbib}
\usepackage[small,bf]{caption}
\usepackage[top=3.5cm,left=3.3cm,right=3.3cm,bottom=4cm]{geometry}

\setlist{noitemsep} 
 
\smallskip 
\titlespacing{\section}{0cm}{1.5cm}{-0.2cm}

\newtheoremstyle{plain}
{12pt}
{9pt}
{\it}
{}
{\bfseries}
{}
{\newline}
{}

\theoremstyle{plain}
\newtheorem{Prop}{Proposition}
\newtheorem{DefAlt}[Prop]{Definition}

\newtheorem{Thm}[Prop]{Theorem}

\begin{document}
\title{Equilibrium notions and framing effects}
\author{Christian Hilbe}
\date{March 2011}
\maketitle
\onehalfspacing

~\\[1cm]
\subsection*{Abstract}
Empirical evidence suggests that the Nash equilibrium leads to inaccurate predictions of human behaviour in a vast set of games. This observation has promoted the development of new solution concepts like the quantal response equilibrium \citep[QRE, see][]{McKPal95} or evolutionary equilibria that are based on the long-run performance of a strategy \citep{AntTraOhtTarNow09,Oht10}. However, it is well-known that the QRE is subject to framing effects: Duplicating a strategy affects the equilibrium predictions. Here we show that the above mentioned evolutionary equilibria exhibit the same inconsistency. Furthermore, we prove that such framing effects are inevitable if a game theoretic solution concept depends differentiably on the payoffs. As a consequence, we argue that differentiable equilibrium notions, while being of great help in analyzing well-specified games, are unsuitable for theoretical modeling, where it is not clear which payoff matrix gives the true representation of an economic interaction.\\[0.5cm]

\noindent {\it Author:} Christian Hilbe, University of Vienna, Austria. Email: hilbe@evolbio.mpg.de\\
{\it Keywords:} Equilibrium; Framing effect; Axiomatic approach; Impossibility theorem;
\newpage

\section{Introduction} 

For any class of games, static or dynamic, with complete or incomplete information, it is a simple task to create an example where the Nash equilibrium  mispredicts human behaviour, as shown for example in \citet{GoeHol01}. This holds true even for the most simple games with only one rationalizable equilibrium, like in the traveller's dilemma introduced by \citet{Basu94}. Suppose that two travelers, returning home from their vacation, discover that the airline has lost their luggage. The airline asks both travelers independently to make claims for compensation and, in order to prevent excessive claims, determines that only the lower of both amounts will be paid. Additionally, it is announced that if the claims are different, the person with the lower claim obtains some reward $R>1$, whereas the same amount $R$ will be deducted from the other traveler's reimbursement as a penalty. In case that only claims between \$180 and \$300 are accepted, the Nash prediction is straightforward: In order to rake in the reward, it is always optimal to undercut the co-player's claim by one dollar. Consequently, the lower bound of \$180 is the unique equilibrium. While this analysis holds true for any $R>1$, simple intuition suggests that subjects in the laboratory may try to coordinate on a higher claim if $R$ is comparably low.\footnote{\onehalfspacing As Kaushik Basu (1994) puts it, the strategy pair $(\text{''large''},\text{''large''})$ is a Nash equilibrium in ill-defined categories; if a player is told that the other player will choose a large number and if the reward $R$ is neglectable, then the best reply is to choose a large number as well. This explanation bears some similarity with the examples in \citet{CamFeh06}, who describe under which conditions a minority of irrational agents can trigger a majority of rational individuals to mimic the minority's behaviour.} Indeed, this intuition is confirmed by experiments: For $R = 5$, around 80 \% of the subjects opt for the maximum claim; only if $R$ is sufficiently increased, claims approach the Nash equilibrium outcome \citep{GoeHol01}. Seemingly, subjects in these experiments do not strictly stick to best responses and do not necessarily eliminate dominated strategies.

These observations are the starting point for several alternative equilibrium notions. In this article we will review two distinct examples, the quantal response equilibrium (QRE) of \citet{McKPal95} and the evolutionary equilibrium described in \citet{Oht10}. Instead of considering  traditional steplike best response correspondences, these equilibrium notions assume that strategy choices are positively but imperfectly related to payoffs.\footnote{\onehalfspacing The same idea has also been applied to some learning models, for example smooth fictitious play, see \citet{FudLev98}.} As a consequence, also dominated strategies may be played from time to time, which in turn may affect equilibrium behaviour. Ironically, {\it because} these alternative equilibrium notions allow a more realistic description of human behaviour, they also have a serious drawback: These equilibrium notions {\it themselves} are subject to framing effects. Different representations of the same economic situation result in different predictions. In particular, giving a strategy a second alias may affect the position of the equilibrium. 
 
We proceed as follows: In the next section, we review the QRE and the evolutionary equilibrium described in \citet{Oht10}. We show how two seemingly equivalent games can lead to diametrically opposed equilibrium predictions. While such framing effects are well-known in the case of the QRE, they have not been previously reported for the evolutionary equilibria. In Section~\ref{SecImpResCh5} we give an unexpected sufficient condition for such framing effects: If an equilibrium concept depends differentiably on the payoffs then inconsistencies are inevitable.\footnote{Roughly speaking, differentiability means that small changes of the payoffs lead to a small and predictable change of the equilibrium. Note that the Nash equilibrium concept does not satisfy this condition, since small changes in the payoffs may completely change the best response correspondences.} As a consequence it is argued in Section~\ref{SecDiscCh5} that the QRE and other differentiable equilibrium notions, although being of great help in analyzing already specified strategic games, might be unsuitable for doing theory, where the true representation of an economic problem is far from being clear.

\section{Examples of equilibrium notions with framing effects} \label{Sec2Ch5}

\subsection{The quantal response equilibrium}

The QRE was introduced by Richard D. McKelvey and Thomas R. Palfrey, first for games in normal form (1995) and later also for extensive form games (1998). \citet{GoeHolPal05} provide an axiomatic foundation. Since then, this concept was applied to various economic settings, including the traveler's dilemma \citep{CapEtAl99} or coordination games \citep{AndGoeHol01}. Typically, the QRE outplays the Nash equilibrium by far when it comes to predict human behaviour in laboratory experiments.\footnote{The overwhelming success is illustrated by the following quote of \citet{CamEtAl04}: {\it Quantal response equilibrium, a statistical generalization of Nash, almost always explains the deviations from Nash and should replace Nash as the static benchmark to which other models are routinely compared.}} Remarkably, the QRE can also be used to estimate the rationality of the subjects \citep{McKPal95} and to which extent they believe in their co-player's rationality \citep{Wei03}.

For our purposes, it will be sufficient to consider the simplest case, a finite normal form game between two players. The $R$-player chooses a row of the 
matrix $M=\big(a_{kl},b_{kl}\big)$, whereas the $C$-player chooses a column. As usual, players are allowed to randomize between their pure actions; we denote by $p^R$ and $p^C$ the respective mixed strategy vectors. For each player $K\in\{R,C\}$, we denote by $u^K_i$ the expected payoff of $K$'s pure action $i$, which of course depends on the co-player's strategy $p^{-K}$, that is $u^K_i=u^K_i(p^{-K})$. A main aspect of the QRE is that choice probabilities are positively but imperfectly related to payoffs. According to the most commonly used parametrization of the QRE, the logit rule, the probability to play action $i$ is determined by the following stochastic reaction function $\sigma$: 
\begin{equation}
p^K_i=\sigma(u^K_i)=\frac{\exp(\lambda \cdot u^K_i)}{\sum_j \exp(\lambda \cdot u^K_j)}
\end{equation}

The sum in the denominator ensures that the probabilities sum up to one. The parameter $\lambda$ can be interpreted as a measure of rationality: $\lambda=0$ means that actions are chosen randomly from the set of possible alternatives, whereas for large $\lambda$ the choice is increasingly biased towards the strategy with the highest payoff. Note that as long as $\lambda<\infty$, even dominated strategies get a positive weight. For analyzing data, the parameter $\lambda$ is typically estimated using the maximum likelihood method. A logit equilibrium is then defined as a fixed point of the map $\sigma$: A pair of mixed strategies $\hat{p}=(\hat{p}^R,\hat{p}^C)$ is an equilibrium if for both players $K \in \{R,C\}$ and all their strategies $i$ the following condition holds:

\begin{equation}
\hat{p}_i^K=\sigma \big(u^K_i (\hat{p}^{-K})\big).
\end{equation}

Such equilibria always exist but need not to be unique. As $\lambda$ goes to infinity, logit equilibria approach Nash equilibria. Furthermore, the graph of 
all fixed points $\hat{p}$ contains a unique branch, starting at the centroid of the strategy simplex for $\lambda=0$ and converging to a unique Nash equilibrium as $\lambda$ approaches infinity, implying that the logit equilibrium can be applied to the problem of equilibrium selection. Since the stochastic reaction function $\sigma$ depends differentiably on the payoffs for $\lambda<\infty$, by the implicit function theorem the same holds true for each branch of the graph of the logit equilibria.\footnote{As we will see in Section~\ref{SecImpResCh5}, the smooth dependence on the payoffs plays a key role. It is valid not only for the logit equilibrium but for the QRE in general, since stochastic reaction functions are generally assumed to be differentiable, see \citet{McKPal95} resp. \citet{GoeHolPal05}.}

Let us illustrate the logit equilibrium with an example taken from \citet{GoeHol01}. Consider the following coordination game in which players receive \$1.80 for coordinating on the high equilibrium and \$0.90 if they coordinate on the low equilibrium. Additionally, the column-player has an outside option that guarantees a safe payoff of \$0.40: 

$$
\begin{array}{l|c|c|c|}
\multicolumn{1}{c}{ }&\multicolumn{1}{c}{~~~~~L~~~~~}	&\multicolumn{1}{c}{~~~~~H~~~~~}	&\multicolumn{1}{c}{~~~~~S~~~~~}\\ \cline{2-4}
L 	&90,90	&0,0	&x,40\\  \cline{2-4}
H 	&0,0	&180,180	&0,40\\ \cline{2-4}
\end{array}
$$

This game has two pure Nash equilibria, $(H,H)$ and $(L,L)$, and the safe option is never part of an equilibrium. Nevertheless, as shown in \citet{GoeHol04}, the outside option has a deciding influence on coordination behaviour in behavioural experiments. In particular, the exact value of $x$ controls which strategies are chosen, with sufficiently low values of $x$ prefering the $(H,H)$ equilibrium. Such an effect is correctly predicted by the logit equilibrium but not by the Nash equilibrium, see Fig.~\ref{FigQRECh5}a for an example with $x=160$: The unique branch of the logit equilibrium, starting in the center for $\lambda=0$ converges to the high equilibrium in the limit of rational agents, $\lambda \rightarrow \infty$. To illustrate that the QRE is subject to framing effects, we consider the same game, but with the second player having two (identical) outside options: 

\begin{equation} \label{PayMat1Ch5}
\begin{array}{l|c|c|c|c|}
\multicolumn{1}{c}{ }&\multicolumn{1}{c}{~~~~~L~~~~~}	&\multicolumn{1}{c}{~~~~~H~~~~~}	&\multicolumn{1}{c}{~~~~~S_1~~~~~}&\multicolumn{1}{c}{~~~~~S_2~~~~~}\\ \cline{2-5}
L 	&90,90	&0,0	&x,40	&x,40\\  \cline{2-5}
H 	&0,0	&180,180	&0,40	&0,40\\ \cline{2-5}
\end{array}
\end{equation}

\begin{figure}[t!]
\centering
\includegraphics[height=4cm]{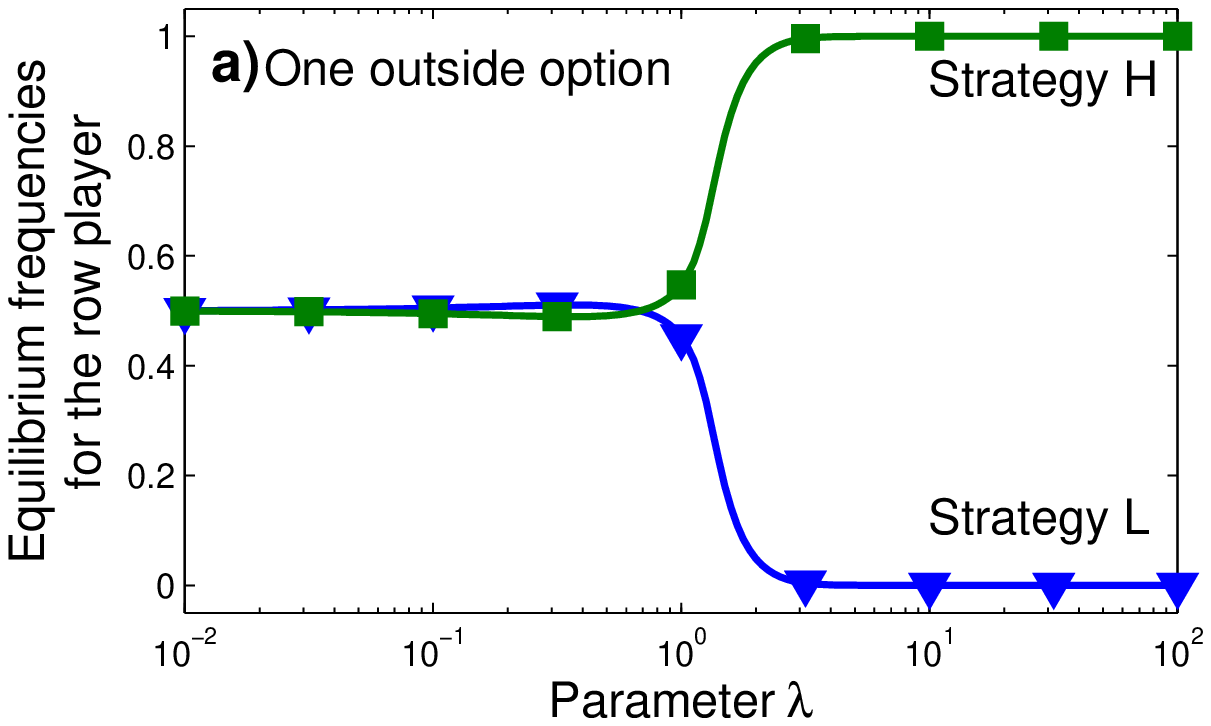}
~~
\includegraphics[height=4cm]{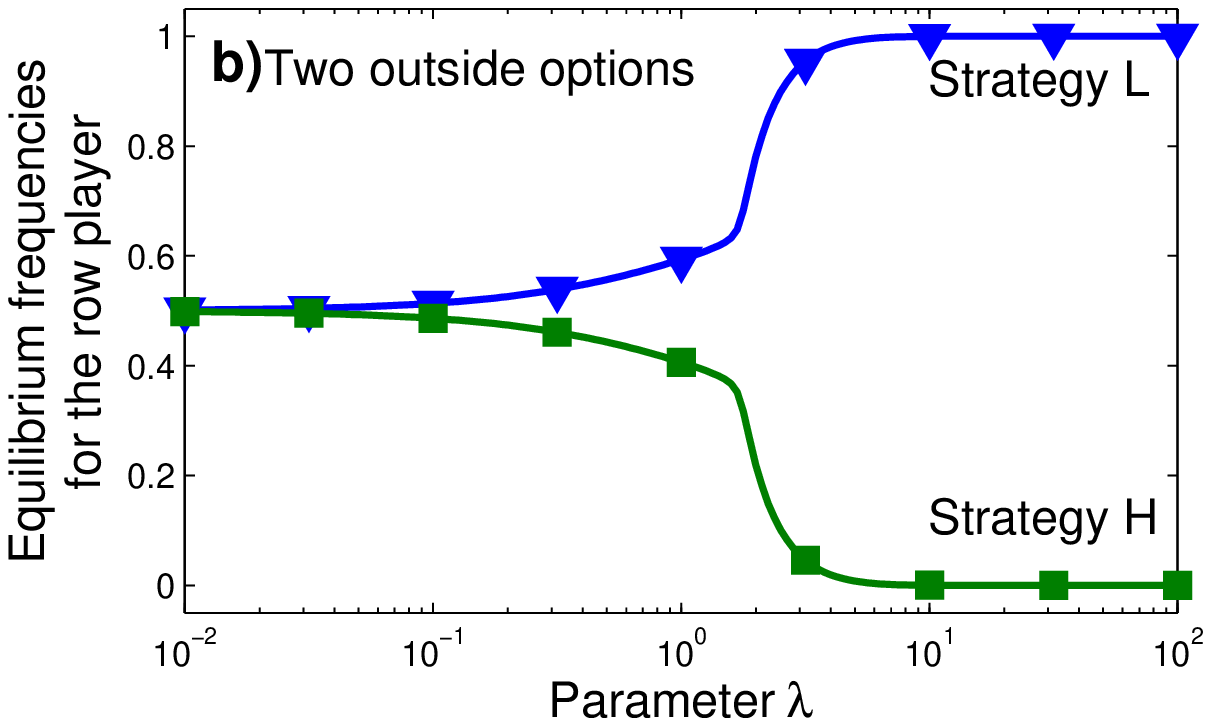}
\parbox{15cm}{\centering \parbox{13cm}{~\\[-.3cm]
\caption{The unique branch of the logit equilibrium for the coordination game with one, respectively with two outside options ($x=160$). a) In the game with one outside option, the value of $x$ is small enough to enable subjects to coordinate on the high equilibrium. b) However, if the problem is represented differently, splitting the safe option in two, players are predicted to coordinate on the low equilibrium.}\label{FigQRECh5}}}
\end{figure} 

While this additional strategy has no effect on the Nash equilibria, it alters the set of logit equilibria (Fig.~\ref{FigQRECh5}b): Giving the outside option a second {\it name} leads the logit equilibrium to select the low instead of the high equilibrium. Thus, the prediction of the logit equilibrium depends sensitively on the exact formulation of the alternatively possible strategies.\footnote{In the above example one might argue that the inconsistency can be avoided if identical columns are omitted by definition. However, if the game is marginally modified such that there are no identical columns, elimination of the additional column seems unjustified:
\begin{equation} \label{PayMat2Ch5}
\begin{array}{l|c|c|c|c|}
\multicolumn{1}{c}{ }&\multicolumn{1}{c}{~~~~~L~~~~~}	&\multicolumn{1}{c}{~~~~~H~~~~~}	&\multicolumn{1}{c}{~~~~~S_1~~~~~}&\multicolumn{1}{c}{~~~~~S_2~~~~~}\\ \cline{2-5}
L 	&90,90	&0,0	&x,40	&x+\varepsilon,40\\  \cline{2-5}
H 	&0,0	&180,180	&0,40	&~~~0,40\\ \cline{2-5}
\end{array}
\end{equation}
More fundamentally, it is typically not the subjects who construct payoff matrices to help them with their decisions, but it is the researcher who uses such tools to describe the decision maker's behaviour. How should one decide which matrix gives the true representation of the decision problem? Similarly, to adapt the argumentation of \citet{KohMer86} on a related issue, {\it elementary transformations}, [like giving a column a second alias], {\it are irrelevant for correct decision making: after all, the transformed matrix is just a different representation of the same decision problem, and decision theory should not be misled by representation effects. To hold the opposite point of view is to admit that decision theory is useless in real-life applications, where problems present themselves without a special formalism.}}

\subsection{Evolutionary equilibria}

\citet{Oht10} considers the following model of an evolutionary dynamics for an asymmetric game: There are two populations, a population $R$ of row-players and a population $C$ of column-players, with population size $N_R$ resp. $N_C$. Each player in $R$ chooses a row $i$ of the matrix $M=( a_{kl},b_{kl})\in \mathbb{R}^{mn}$, whereas each player in $C$ chooses a column $j$. Then, every subject in population $R$ plays against every subject in the other population, leading to the payoffs $u_i^R$ and $u_j^C$, respectively. Subsequently, the fitness $f_i$ of a player with strategy $i$ is defined by an exponential transformation of its payoff, i.e. $f_i^K=\exp(\delta u_i^K)$ with $K\in \left\{ R,C \right\}$. The factor $\delta$ measures the importance of the game for the fitness of a player and is usually called the strength of selection. If $\delta\rightarrow 0$, each agent has approximately the same fitness, a case which is termed the weak selection limit.

After those interactions, one subject (of any of the two populations) is chosen at random. This agent is allowed to change its strategy by imitating the strategy of another player of the same population. It is assumed that strategies with higher fitness are more likely to be adopted. More specifically, if $N_k^R$ denotes the current number of row-players with strategy $k$, then the probability that a randomly chosen row-player imitates an agent with strategy $i$ is given by 
\begin{equation}
p_i^R=\frac{N_i^Rf_i^R}{\sum_k N_k^Rf_k^R}.
\end{equation} 
Additionally, one allows mutations: With probability $u$, the agent does not imitate others, but chooses randomly any of the available strategies.

Overall, this evolutionary dynamics results in a stochastic selection-mutation process without absorbing states. In some special cases, it is possible to explicitly calculate the invariant distribution of the process. One important case is the limit of weak selection, in which the fitness of each individual is largely independent of its payoff in the game. As a consequence, each strategy for population $R$ is approximately played with probability $1/m$, only slightly truncated by a term $\varphi_i$, which reflects the impact of the respective strategy. If both populations are of equal size this term is, up to a multiplicative constant, given by
\begin{equation} \label{phi}
\varphi_i=\bar{a}_i-\bar{a},
\end{equation}
where $\bar{a}_i$ denotes the average of all feasible payoffs for a player with strategy $i$, that is $\bar{a}_i=\sum_ja_{ij}/n$, and $\bar{a}$ denotes the average of all feasible payoffs for individuals in population $R$, $\bar{a}=\sum_{i,j}a_{ij}/(mn)$.\footnote{This expression for $\varphi_i$ resembles the well-known replicator dynamics, where it is assumed that the frequency of players with strategy $i$ increases if the payoff $u_i$ exceeds the average payoff in the population $\bar{u}$ \citep[see, for example][]{Wei95}. However, while the payoffs $u_i$ and $\bar{u}$ may vary over time, depending on the current state of population, the term $\varphi_i$ is constant and does only depend on the payoff matrix.} Since this mutation-selection process does not settle down on any stable state, the deviation terms $\varphi_i$ take the role of the major characteristic of the system. It is said that {\it selection favors strategy $i$} if $\varphi_{i}$ is positive \citep{Now06}. Furthermore, one may compare two different strategies with each other: \citet{AntTraOhtTarNow09} call {\it strategy $i$ more abundant than strategy $k$} if $\varphi_i>\varphi_k$. In effect, this approach allows a ranking of the strategies - based on the long run performance of each strategy in the above described evolutionary process. 

However, it is easy to show that this evolutionary equilibrium exhibits the same framing effects as the QRE. In fact, calculating $\varphi$ for the example in the previous section (for $x=60$) yields $\varphi_L^R=-5$ and $\varphi_H^R=5$ in the case of representation (\ref{PayMat1Ch5}), respectively $\varphi_L^R=3.75$ and $\varphi_H^R=-3.75$ in the case of representation (\ref{PayMat2Ch5}).
Hence, neither does the absolute value of $\varphi_L^R$ allow a consistent assessment across the different treatments, nor is the order of $H$ and $L$ left invariant. In the case of weak selection and uniform mutations, one can easily determine the reason for this inconistency: While in the first representation, the outside option is played by roughly $1/3$ of all column-players, this fraction increases to approximately $1/2$ of the $C$-population if there are two outside options, which in turn encourages row-players to choose strategy $L$.

\section{An Impossibility Result}\label{SecImpResCh5}

Let us turn to the question whether it is possible to construct other equilibrium notions (or other parametrizations of the QRE, respectively more general evolutionary equilibria) that avoid the inconsistencies shown in the previous section. For simplicity, we focus on two-player normal form games and identify each game $\Gamma$ with its payoff matrix. For this reason, denote by $\mathcal{M}=\big\{(a_{ij},b_{ij}):~a_{ij},b_{ij}\in\mathbb{R}\big\}$ the set of all payoff matrices for normal form games, and let $\mathcal{M}_k\subset \mathcal{M}$ be the set of all payoff matrices that have exactly $k$ rows (i.e., that admit $k$ pure strategies for the row-player). 

\begin{DefAlt}[Assessment formula]
A game theoretic assessment is a function $f:\mathcal{M}_k \rightarrow \mathbb{R}^k.$
\end{DefAlt}

One may interpret each entry of $f(M)$ as the predicted equilibrium frequencies  for the row-player in the normal form game defined by the payoff matrix $M$.\footnote{In this case one can restrict the image of the game theoretic assessment to the unit simplex $\Delta^k$ instead of $\mathbb{R}^k$. Since equilibria need not to be unique, this interpretation requires that for each game one particular equilibrium is selected out of the set of possible equilibria.} In particular, note that the above definition of game theoretic assessments includes the logit equilibrium $\hat{p}$ and the evolutionary assessment $\varphi$ from the previous section.

In order to exclude framing effects, we demand that equivalent representations of a game lead to the same assessments. Up to a renumbering of the strategies of the column-player, we say that two matrices $M,\hat{M} \in \mathcal{M}_k$ are equivalent if they result in the same matrix after deleting all columns that are a copy of a previous column. More formally, $M \sim \hat{M}$ if for all columns $j$ of M there is a column $l$ in $\hat{M}$ such that $(a_{ij},b_{ij})=(\hat{a}_{il},\hat{b}_{il})$ for all rows $i$ (and vice versa, for all columns $l$ in $\hat{M}$ there is such a column $j$ in $M$). Obviously, this defines an equivalence relation on the set $\mathcal{M}_k$ for all $k$. 

\begin{DefAlt}[Consistent assessment formulas]
Fix a $k\ge 2$. An assessment formula $f:\mathcal{M}_k \rightarrow \mathbb{R}^k$ is called consistent if it has the following properties:
\begin{enumerate}
\item {\it Non-manipulability:} If $M,K\in \mathcal{M}_k$ and $K \sim M$ then $f(K)=f(M)$.\footnote{A similar condition can be found, for example, in Milnor's famous work on {\it games against nature}, see \citet{Mil51}.}
\item {\it Validity:} If the row-player's strategy $i$ is strictly dominated, then $i$ cannot be optimal, $f_i(M)<\max_j f_j(M)$.
\end{enumerate}
\end{DefAlt}

Non-manipulability means that a consistent assessment is {\it well-defined} with respect to the above equivalence relation, i.e. it respects that two matrices $M,\hat{M}$ with $M \sim \hat{M}$ represent the same game. Validity excludes constant assessment formulas from being considered. After these preparations, we are able to formulate the main result:

\begin{Thm}[An impossibility theorem]
There is no assessment fromula that is both, consistent and differentiable.
\end{Thm}
\begin{proof}
We show the case $k=2$: Suppose there is such a consistent and differentiable assessment $f$ and consider the arbitrary but fixed matrix $$M=
\left[
\begin{array}{cc}
(a_{11},b_{11})	&(a_{12},b_{12})\\
(a_{21},b_{21})	&(a_{22},b_{22})
\end{array}•
\right].$$
Let $\partial f/\partial \alpha_{ij}$ denote the marginal change of the assessment if the row-player\rq{}s payoff in the $i$-th row and the $j$th column is varied. The idea of the proof is as follows: Non-manipulability implies that all partial derivatives $\partial f/ \partial \alpha_{ij}(M)$ are zero, which suggests that the value of $f(M)$ is independent of the values of $a_{ij}$. This in turn contradicts validity. To show that the derivatives equal zero, we define the two matrices
$$M_1(t)=
\left[
\begin{array}{cc} 
(a_{11}+t,b_{11})	&(a_{12},b_{12})\\
(a_{21},b_{21})	&(a_{22},b_{22})
\end{array}•
\right]\text{~~~and}$$ 
$$M_2(t)=
\left[
\begin{array}{cccc}
(a_{11}+t,b_{11})	&(a_{11}+t,b_{11})	&(a_{11}+t,b_{11}) &(a_{12},b_{12})\\
(a_{21},b_{21})	&(a_{21},b_{21})	&(a_{21},b_{21})	&(a_{22},b_{22})
\end{array}
\right].$$
$M_2(t)$ is obtained from $M_1(t)$ by doubling the first column two times. Note that $M_1(0)=M$. Next we define two functions that measure how the respective strategy assessments vary with $t$, i.e. we define $u_i:\mathbb{R} \rightarrow \mathbb{R}^k$ with $u_i(t)=f\big(M_i(t)\big)$ for $i=1,2.$ Since $M_1(t) \sim M_2(t)$ for all $t$, non-manipulability implies that $u_1(t)=u_2(t)$. In particular, the derivatives for $t=0$ coincide:
\begin{equation} \label{dfda11}
u_1\rq{}(0)=\frac{\partial f}{\partial \alpha_{11}}(M)=
\frac{\partial f}{\partial \alpha_{11}}\big(M_2(0)\big) + 
\frac{\partial f}{\partial \alpha_{12}}\big(M_2(0)\big) + 
\frac{\partial f}{\partial \alpha_{13}}\big(M_2(0)\big)=u_2\rq{}(0).
\end{equation}
Therefore, since we want to show $\partial f/\partial \alpha_{11}(M)=0$, we have to compute the expression on the right hand\rq{}s side of (\ref{dfda11}). For this reason, we define two new matrices: 

$$M_3(t)=
\left[
\begin{array}{cccc}
(a_{11},b_{11})	&(a_{11},b_{11})	&(a_{11}+t,b_{11}) &(a_{12},b_{12})\\
(a_{21},b_{21})	&(a_{21},b_{21})	&(a_{21},b_{21})	&(a_{22},b_{22})
\end{array}•
\right]\text{~~~and}$$ 
$$M_4(t)=
\left[
\begin{array}{cccc}
(a_{11},b_{11})	&(a_{11}+t,b_{11})	&(a_{11}+t,b_{11}) &(a_{12},b_{12})\\
(a_{21},b_{21})	&(a_{21},b_{21})	&(a_{21},b_{21})	&(a_{22},b_{22})
\end{array}•
\right].\text{~~~~~}$$
Note that these two matrices have the same reduced normal form and hence are equivalent. Additionally, they fulfill $M_3(0)=M_4(0)=M_2(0)$. If we again define functions $u_i(t)=f\big(M_i(t)\big)$ for $i=3,4$, we may conclude that
\begin{equation} 
u_3\rq{}(0)=\frac{\partial f}{\partial \alpha_{13}}\big(M_2(0)\big)=
\frac{\partial f}{\partial \alpha_{12}}\big(M_2(0)\big) + 
\frac{\partial f}{\partial \alpha_{13}}\big(M_2(0)\big)=u_4\rq{}(0),
\end{equation}
\noindent and therefore $\partial f/\partial \alpha_{12}\big(M_2(0)\big)=0$. With a similar calculation one can show that the other two expressions on the right hand\rq{}s side of (\ref{dfda11}) , $\partial f/\partial \alpha_{11}\big(M_2(0)\big)$ and $\partial f/\partial \alpha_{13}\big(M_2(0)\big)$, vanish as well. Therefore, we indeed end up with  $\partial f/\partial \alpha_{11}(M)=0$. A symmetry argument then immediately implies that $\partial f/\partial \alpha_{ij}(M)=0$ for all $i$ and $j$. As a consequence, the assessment $f(M)$ does not vary in the row-player\rq{}s payoffs, which leads to a contradiction with the validity of the assessment.
\end{proof}

Therefore, we must conclude that there is no reasonable equilibrium concept that is both, non-manipulable and smooth. If we interpret the value of $f_i(M)$ slightly differently, as an indicator of the performance of strategy $i$, then the previous theorem states that is impossible to measure the success of a strategy with a differentiable formula.

In particular, the inconsistencies of the evolutionary assessment $\varphi$ cannot be simply attributed to the assumption of weak selection. Even in the case of some positive but finite selection pressure $\delta$, the stationary distribution in \citet{Oht10} depends differentiably on the entries of the payoff matrix.

\section{Discussion} \label{SecDiscCh5}

Explaining human behaviour with game theoretic models faces at least two difficulties. Firstly, the modeller does usually not know the exact subjective utilities of the agents; instead there might be only some rough estimates. In order to obtain robust results, one might therefore require that the output of the model depends differentiably on the input data. Secondly, in order to set up the model, the researcher needs to choose one specific description of reality, out of many alternatively possible descriptions. One such choice might entail, for example, to determine whether a certain player has only one outside option or several similar options. In the best case - if the methods are consistent in the sense defined above - the exact representation of the game does not affect the qualitative results. 

However, as we have shown, the two requirements of consistency and differentiability are incompatible. If the results of a game theoretic equilibrium notion depend differentiably on the payoffs, then these results also depend on the representation. In this sense, solution concepts for games in strategic form are necessarily imperfect. Therefore any game theoretic concept that can be applied to normal form games faces the choice whether it violates one requirement or the other. The matrix presented as Tab.~\ref{ClassGameCh5} attempts to give an overview over some choices that were made. It classifies some popular game theoretic tools according to whether they violate the smooth dependence on payoffs condition or the non-manipulability condition. Of course, such a list is notoriously incomplete and each cell of this matrix might contain several other elements - with the exception of the cell that corresponds to the differentiable and consistent concepts.

\begin{table}[htp]
\definecolor{hellgrau}{gray}{0.9}
\definecolor{grau}{gray}{0.8}
\centering
\begin{tabular}{lcc}
	&Differentiable concepts	&Non-differentiable concepts\\[0.2cm]

 \parbox{.2cm}{ \centering 
\begin{turn}{90} \parbox{4.5cm}{\centering Consistent concepts} \end{turn}}
&\colorbox{hellgrau}{\parbox{7cm}{~\\[6cm]}}
&\colorbox{grau}{
\parbox{7cm}{\footnotesize 
Nash equilibrium \citep{Nash50}\\[0.3cm]
{\bf Refinements of the Nash equilibrium}\\[0.1cm]
\text{~~~} Perfect equilibrium \citep{Selt75}\\
\text{~~~} Proper equilibrium \citep{Myer78}\\[0.3cm]
{\bf Approaches that apply the Nash equilibrium to transformed utilities}\\[0.1cm]
\text{~~~} Fairness model of \citet{FehSch99}\\[0.3cm]
{\bf Learning processes for which Nash equilibria are rest points}\\[0.1cm]
\text{~~~} Fictitious play \citep{Brown51}\\
\text{~~~} \parbox{5.5cm}{Replicator dynamics \citep{TayJon78}}\\
\text{~~~} \parbox{5.5cm}{Best response dynamics \citep{GilMat91}}}}\\
~	&~	&~\\

 \parbox{.2cm}{\centering 
\begin{turn}{90} \parbox{4.5cm}{\centering Inconsistent concepts} \end{turn}}
&\colorbox{grau}{
\parbox{7cm}{\footnotesize 
\parbox{6.8cm}~\\[0.3cm]{\bf Behavioural equilibrium notions}\\[0.1cm]
\text{~~~} QRE \citep{McKPal95}\\
\text{~~~} \parbox{6cm}{Level-$k$ reasoning model \citep{StaWil95}}\\
\text{~~~} Noisy introspection \citep{GoeHol04}\\[0.3cm]
{\bf Smooth learning processes}\\[0.1cm]
\text{~~~} \parbox{6.5cm}{Exponential fictitious play \citep{FudLev98}}\\[0.3cm]
\parbox{6.8cm}{\bf Long run equilibria for evolutionary processes with uniform mutations and smooth selection}\\[0.1cm]
\text{~~~} \parbox{6cm}{Moran process \citep{AntTraOhtTarNow09,Oht10}}\\[0.3cm]}}
&\colorbox{hellgrau}{
\parbox{7cm}{\footnotesize 
~\\[2.4cm]
\parbox{6.8cm}{\bf Long run equilibria for evolutionary processes with uniform mutations and best-reply selection}\\[0.1cm]
\text{~~~} \parbox{6cm}{Moran process with strong selection \citep{FudNowTayImh06}}\\[2.1cm]}}
\end{tabular}
\parbox{15cm}{\centering \parbox{13cm}{~\\[-.3cm]
\caption{A classification of game theoretic concepts}\label{ClassGameCh5}}} 
\end{table}

A natural question is then to ask which of the two requirements is the more indispensable one. Differentiable equilibrium notions, and in particular the QRE, are quite successful in predicting human behaviour for normal form games - once it is known which representation of the game the subjects choose. In laboratory experiments this is certainly no problem, since it may be assumed that the players' internal model of the game is close to the instructions that are provided by the experimenter (in particular it is likely that all subjects have a similar internal representation). From a behavioural point of view, the framing effects presented in the previous sections even seem to be a desirable feature - after all it is well documented that humans are subject to framing effects as well.\footnote{For the related question whether subjects in dynamic games play differently if confronted with different game trees that represent formally equivalent games, see \citet{McKPal98}.} Psychologically, it is not unreasonable to expect that a duplication of the outside option increases the number of $L$ players in game (\ref{PayMat1Ch5}). The outside options may act as a coordination device: Because both options point to the low equilibrium, this equilibrium may be interpreted as a focal point \citep{Schelling60}. 

However, if it comes to explain human behaviour in the field it is not at all clear how individuals perceive their interactions, let alone that these perceptions are comparable across subjects. For theoretical modeling, the above described framing effects are undesirable (or even dangerous). If an equilibrium concept leads to predictions that depend on the representation of the game (which is chosen by the modeler himself), then the results will be somewhat arbitrary in the best case and manipulable in the worst.

A possible solution to avoid framing effects in differentiable equilibrium notions is to consider the equivalence class of a game. That is, instead of calculating the logit equilibrium $\hat{p}_\lambda(M)$ of a game $M$ one may calculate the set of possible logit equilibria $\hat{p}_\lambda([M])$ for all games that are equivalent to $M$,
\begin{equation}
[M]=\left\{\hat{M}\in \mathcal{M}_k~\big|~\hat{M}\sim M\right\}.
\end{equation}
However, in this case, the logit equilibrium loses its ability to select a unique Nash equilibrium in the limit of rational agents, $\lambda\rightarrow \infty$. Instead, most of the Nash equilibria of a game $M$ (including all strict Nash equilibria) are predictable by the unique branch of $\hat{p}_\lambda(\hat{M})$ - if only the game is appropriately framed. Therefore, it seems to me that the solution concept of the Nash equilibrium is \citep[almost, see][]{Selt75,Myer78} as good as it gets.

\bibliographystyle{abbrvnat} \bibliography{literature}
\end{document}